\documentclass[journal,twoside,web]{ieeecolor}
\usepackage{lcsys}
\usepackage{cite}
\usepackage{amsmath,amssymb,amsfonts}
\usepackage{graphicx}
\usepackage{textcomp}
\def\BibTeX{{\rm B\kern-.05em{\sc i\kern-.025em b}\kern-.08em
    T\kern-.1667em\lower.7ex\hbox{E}\kern-.125emX}}
\usepackage{subcaption}
\usepackage{algpseudocode}
\usepackage{algorithmicx} 
\usepackage{algorithm}
\usepackage{varwidth}
\usepackage{cite}
\usepackage{color}
\usepackage{multirow}

\newtheorem{lemma}{Lemma}

\newtheorem{proposition}{Proposition}
\newtheorem{theorem}{Theorem}

\begin{document}
\title{Modeling and Designing Non-Pharmaceutical Interventions in Epidemics:\\ A Submodular Approach}
\author{Shiyu Cheng, \IEEEmembership{Graduate Student Member, IEEE}, Luyao Niu, \IEEEmembership{Member, IEEE}, Bhaskar Ramasubramanian, \IEEEmembership{Member, IEEE}, Andrew Clark, \IEEEmembership{Senior Member, IEEE}, and Radha Poovendran, \IEEEmembership{Fellow, IEEE}
\thanks{S. Cheng and A. Clark are with the Department of Electrical and Systems Engineering, Washington University in St. Louis, St. Louis, MO, USA. Email: \{cheng.shiyu, andrewclark\}@wustl.edu}
\thanks{L. Niu is with the Department of Electrical and Computer Engineering, University of Washington, Seattle, WA. Email: luyaoniu@uw.edu}
\thanks{R. Poovendran is with the Network Security Lab, Department of Electrical and Computer Engineering, University of Washington, Seattle, WA. Email: rp3@uw.edu}
\thanks{B. Ramasubramanian is with the Electrical and Computer Engineering Program, Western Washington University,
Bellingham, WA, USA. Email: ramasub@wwu.edu}
\thanks{This research was supported by the AFOSR grants FA9550-22-1-0054 and FA9550-23-1-0208, and NSF grant 2153136.}
}

\maketitle
\thispagestyle{empty}
\begin{abstract}
This paper considers the problem of designing non-pharmaceutical intervention (NPI) strategies, such as masking and social distancing, to slow the spread of a viral epidemic. We formulate the problem of jointly minimizing the infection probabilities of a population and the cost of NPIs based on a Susceptible-Infected-Susceptible (SIS) propagation model. 
To mitigate the complexity of the problem, we consider a steady-state approximation based on the quasi-stationary (endemic) distribution of the epidemic, and prove that the problem of selecting a minimum-cost strategy to satisfy a given bound on the quasi-stationary infection probabilities can be cast as a submodular optimization problem, which can be solved in polynomial time using the greedy algorithm.
We carry out experiments to examine effects of implementing our NPI strategy on propagation and control of epidemics on a Watts-Strogatz small-world graph network. We find the NPI strategy reduces the steady state of infection probabilities of members of the population below a desired threshold value. 
\end{abstract}

\begin{IEEEkeywords}
Epidemics, biological systems, submodular optimization, networked systems.
\end{IEEEkeywords}

\section{Introduction}
\label{sec:intro}
\IEEEPARstart{V}{iral} epidemics have affected millions of people throughout human history, leading to immense loss of life and economic damage.
The need to understand and contain epidemics has spurred the development of mathematical models of disease propagation, such as the Susceptible-Infected-Susceptible (SIS) and Susceptible-Infected-Recovered (SIR) models \cite{pare2020modeling}. These models have also informed strategies for limiting the spread of the epidemic, e.g., by allocating scarce doses of vaccines to the most vulnerable populations \cite{nowzari2016analysis}.

There has been significant research interest on control techniques for selecting optimal intervention strategies. Most of these techniques have focused vaccination-based epidemic mitigation, with the goal of allocating vaccines in order to drive the propagation rate below a threshold level and eventually remove the disease from the population \cite{preciado2013optimal}. On the other hand, recent epidemics such as SARS and COVID-19 have shown that \emph{non-pharmaceutical interventions} (NPIs) such as masking and social distancing can also be effective in slowing the spread of disease, especially in situations where vaccines are unavailable \cite{bouchnita2020hybrid, garcia2022assessing, hens2020covid}. 

The problem of choosing an optimal NPI strategy can be formulated within the control-theoretic framework of changing system dynamics and network structure to control epidemic spread \cite{cenedese2021optimal}.
The challenge that arises is that the set of intervention strategies is inherently discrete and combinatorial, since it involves choosing a subset of network nodes or edges to employ NPIs. The resulting set of intervention strategies has exponential complexity in number of nodes/ edges, and will be difficult to implement in practice. 

In this paper, we propose a submodular optimization approach to apply NPIs to jointly minimize infection probability and cost of mitigation. We consider the network to include clusters that group individuals based on geographic location or institutional affiliation. Instead of selecting specific individuals, we select clusters to apply NPIs based on the perspective of a public health authority. We consider an SIS model in which the effect of NPIs is quantified by a reduction in propagation rate. 
To reduce complexity, we consider the system model based on a mean-field approximation. SIS models are known to converge to a quasi-stationary \emph{endemic} distribution when the propagation rate exceeds an epidemic threshold \cite{van2008virus}, and the endemic distribution can be obtained by solving a system of nonlinear equations. 
We formulate the problem of selecting a minimum-cost intervention strategy to ensure that the steady state infection probability is below a desired threshold, and prove that this problem is submodular. We propose an efficient greedy algorithm for the quasi-stationary model that scales to large networks. Our approach is validated via a simulation study of epidemic propagation on small-world networks \cite{watts1998collective}.

The paper is organized as follows. Section \ref{sec:related} reviews the related work. Section \ref{sec:preliminaries} presents the system model and background on submodularity.
Section \ref{sec:endemic} presents the problem formulation and our approach for optimizing the steady state of infection probability using a submodular approach. Section \ref{sec:simulation} contains simulation results. Section \ref{sec:conclusions} concludes the paper.
\section{Related Work}\label{sec:related}

The modeling and analysis of epidemic processes has been a widely studied area of research. 
The surveys in \cite{pare2020modeling, nowzari2016analysis} (and references therein) provide an excellent overview of the most widely used models and analysis techniques in this domain. 
Both works present a variety of epidemic process models, and analyze key properties of each model such as stability; following this, \cite{pare2020modeling} provides solution approaches to estimate
values of parameters associated with these models, while \cite{nowzari2016analysis} presents optimal control solutions to stop epidemic spread. 
While the spirit of our paper is similar to \cite{nowzari2016analysis}, our approach is significantly different. 
Specifically, our solutions leverage a submodular property inherent to the underlying model structure to develop \emph{computationally efficient NPI strategies} to slow the spread of an epidemic. 

In the continuous-time setting, \cite{van2008virus} proposed a networked SIS dynamics model based on the mean-field approximation; see also \cite{fall2007epidemiological}. 
Concurrently, there is a large body of work emphasizing techniques to mitigate the spread of diseases in networked systems. These solutions can be broadly classified into methods which manipulate `edges' and those that manipulate `nodes'. 
Distributed link removal strategies were designed in \cite{xue2018distributed, liu2021distributed} to control the spread of epidemics processes represented by SIS or SIR models. 
On the other hand, \cite{drakopoulos2014efficient} devises a curing strategy that critically relies on strategic elimination of infected nodes in a SIS model to control disease spread. 
They subsequently show that this mechanism can be used to ensure that the epidemic dies out at a chosen rate. 

Several works have examined developing solutions to control epidemic spread when there is a constraint on resources/ costs. 
Particular solution approaches include minimizing steady state values of nodes \cite{gourdin2011optimization}, geometric programming-based solutions \cite{preciado2014optimal}, mitigating disease spread while maximizing recovery speed under a cost \cite{mai2018distributed} or vaccine deployment \cite{preciado2013optimal} budget, and optimal control-based solutions \cite{grandits2019optimal}. 
Insights from model-predictive control were leveraged to formulate disease spread strategies in \cite{kohler2018dynamic, carli2020model, cenedese2021optimal}. 

To the best of our knowledge, there is limited work on the design of computationally efficient NPI strategies.
The closest work related to ours is \cite{watkins2018control}, where the authors use a submodular approach combined with model predictive control to identify infected nodes to be removed. 
The design and use of NPI strategies to control disease spread have been investigated in the context of the recent COVID-19 pandemic in \cite{bouchnita2020hybrid, garcia2022assessing, hens2020covid}. 
The authors of \cite{bouchnita2020hybrid} showed that panic situations exacerbated risk of disease transmission in crowds despite social distancing. 
All of these works presented extensive empirical analysis of data that was available as a consequence of COVID-19. Furthermore, some works analyze the control of epidemic outbreaks based on metapopulation models, in which `nodes' in the network denote large subpopulations \cite{preciado2013traffic, ye2020network}. 
In comparison, in this paper, we consider a network characterized by clustering and provide a comprehensive modeling and analytical framework to select clusters of nodes in a graph to implement NPI strategies in a cost-effective manner. 
Our solutions are grounded on submodular optimization, which enables us to use a greedy optimization algorithm with provable guarantees on performance.  
\section{Preliminaries}
\label{sec:preliminaries}
This section presents notation, the epidemic model that we consider, and necessary background on submodularity.

\subsection{Notations} 
We use $|\bullet|$ to denote the cardinality of a set. A boldface letter, such as $\boldsymbol{x}$, denotes a vector, and $x_{i}$ denotes the $i$-th entry in $\boldsymbol{x}$. We use $\boldsymbol{x}\leq \boldsymbol{y}$ to denote $x_i \leq y_i$ for all $i$ entries in $\boldsymbol{x}$ and $\boldsymbol{y}$. Let $\boldsymbol{x}< \boldsymbol{y}$ indicate $x_{i}<y_{i}$ for all $i$. For matrices $A$, $\rho(A)$ is the spectral radius of $A$. We define $\overline{A} \geq A$ if $\overline{A}_{ij}\geq A_{ij}$ holds for each entry in $A$ and $\overline{A}$.
\subsection{Epidemic Model}
\label{sec:model}
We consider a population of $n$ individuals, indexed in the set $V = \{1,\ldots,n\}$. We say that an edge exists between two individuals $i$ and $j$ if they interact, and let $E$ denote the set of edges. We let $G=(V,E)$ denote a graph with node set $V$ and edge set $E$, which captures the pairwise social interactions in the population. We assume $G$ is a connected undirected graph. We let $N(i) \triangleq \{j : (i,j) \in E\}$ denote the set of neighbors of $i$.
The epidemic propagates according to a Susceptible-Infected-Susceptible (SIS) process, defined as follows \cite{pare2020modeling}. 
The state of a node in the network is either infected or cured. We use two independent Poisson processes to describe the curing and infection of node $i\in V$. The rates for these two Poisson processes are $\gamma_i$ and $\beta_i$, respectively. We assume $\gamma_i>0$ and $\beta_i>0$. A mean-field approximation can be applied to describe the evolution of the infection probability of $i\in V$.
In this approximation, each node $i$ has a state $x_{i}(t)\in [0,1]$ with $\boldsymbol{x}(0)\neq \boldsymbol{0}$, denoting the probability of infection. The evolution of $x_{i}(t)$ can be described as 
\begin{equation}
    \dot{x}_{i}(t)=-\gamma_i x_{i}(t)+(1-x_{i}(t))\beta_{i}\sum_{j\in N(i)}a_{ij}x_j(t).
\end{equation}
Let $X(t) = diag(x_1(t), \ldots, x_n(t))$, $B = diag(\beta_1, \ldots, \beta_n)$, and $D = diag(\gamma_1, \ldots, \gamma_n)$. The dynamics of the entire system are
\begin{equation}
\label{eq:system-dynamics}
    \dot{\boldsymbol{x}}(t) = (BA-D)\boldsymbol{x}(t) - X(t) BA \boldsymbol{x}(t),
\end{equation}
where $A = [a_{ij}]$ denotes an adjacency matrix with $a_{ij}\geq 0$ \cite{pare2020modeling}.
The value of $a_{ij}$ quantifies strength of a connection between $i$ and $j$. We assume $a_{ij} = a_{ji}>0$ if $j\in N(i)$ and  $a_{ij}=0$ if $j\notin N(i)$. For states in \eqref{eq:system-dynamics}, if $x_{i}(0)\in [0,1]$, $\forall i\in V$, then $x_{i}(t)\in [0,1]$ for all $t\geq 0$ \cite{van2008virus}.
The system with dynamics \eqref{eq:system-dynamics} has a trivial equilibrium point, $\boldsymbol{x}=\boldsymbol{0}$, called the \emph{disease-free} state. However, in some cases, system states will converge to another equilibrium point $\boldsymbol{x}^{\ast}> \boldsymbol{0}$, termed an \emph{endemic} state. The existence of the endemic state is dependent on parameters of the system. Let $\mathcal{R}_0 = \rho(D^{-1}BA)$. For system \eqref{eq:system-dynamics}, if $\boldsymbol{x}(0)\neq 0$ and $\mathcal{R}_0\leq 1$, then $\boldsymbol{x}=0$ is an asymptotically stable equilibrium point with region of attraction $[0,1]^{n}$. Otherwise, if $\mathcal{R}_0> 1$, then $\boldsymbol{x}^{\ast}>0$ is globally asymptotically stable \cite[Thm. 2]{khanafer2016stability}.

We next describe the impact of NPIs. Such an intervention aims to lower the strength of the connection between $i$ and $j$ by taking precautions such as wearing masks. Specifically, if individual $i$ and individual $j$ take the NPIs for their interaction, then edge weights between them will decrease to $a_{ij}(1-\theta_2)$. Otherwise, if only one individual among $i$ and $j$ takes the NPI, the edge weights between them will decrease to $a_{ij}(1-\theta_1)$,
where $\theta_1, \theta_2\in (0.5,1)$ denote the effect of NPIs with $\theta_1 <\theta_2$ \cite{ueki2020effectiveness}. 
This implies the effect of NPIs is more significant when both individuals take it than when only one of them applies it.

In the present paper, we assume the network includes several clusters of nodes. The NPI strategy is to select clusters and let the individuals in these clusters take the NPIs during the interaction with their neighbors. For example, clusters may represent individuals who are from the same geographic area or affiliated with the same firm or university. Assuming those clusters are indexed in the set $S=\{1, \ldots, M\}$, we consider the perspective of a public health authority who must 
 select a subset of clusters $\overline{S}\subseteq S$ to apply the NPI strategy. Note that the case where each node $i$ can be individually incentivized to use NPIs by the health authority corresponds to the case where each cluster is a singleton set $\{i\}$. We use $V_{r}\subseteq V$ to denote the nodes in cluster $r\in S$. When we apply NPIs on the clusters set $\overline{S}$, the corresponding selected nodes set is defined as $\overline{V}(\overline{S})=\cup_{r\in \overline{S}}V_r$. When applying NPIs, the dynamics of node $i$ change to 
 \begin{equation}
 \label{eq:CT-node-dynamics}
    \dot{x}_{i}(t)=-\gamma_i x_{i}(t)+(1-x_{i}(t))\beta_{i}\sum_{j\in N(i)}a_{ij}(\overline{S})x_j(t),
\end{equation}
where $a_{ij}(\overline{S})$ is defined by
\begin{align*}
        &a_{ij}(\overline{S})=\left\{
        \begin{array}{ll}
        a_{ij}, & i,j\notin \overline{V}(\overline{S})\\
        a_{ij}(1-\theta_2), & i,j\in \overline{V}(\overline{S}) \\
        a_{ij}(1-\theta_1), & \mbox{otherwise}
        \end{array}
        \right.   
\end{align*}
Therefore, the dynamics of the entire system are 
\begin{equation}
\label{eq:DT-system-dynamics}
     \dot{\boldsymbol{x}}(t) = (BA(\overline{S})-D)\boldsymbol{x}(t) - X(t) BA(\overline{S}) \boldsymbol{x}(t),
\end{equation}
where $A(\overline{S}) = [a_{ij}(\overline{S})]$.

We next define the cost of applying the NPIs. This cost represents the economic cost and the desire of individuals to maintain regular social interaction based on different scenarios.
We consider three metrics with the corresponding cost of the node and the total NPI cost as follows. 
\begin{itemize}
    \item Additive cost: Let $c_{r,1}$ denote the financial incentive paid for applying NPI  to each individual in cluster $r$. If $i$ is included in multiple selected clusters, it will receive funds from all those clusters. 
    The cost of node is $\tilde{c}_{i}(\overline{S})=\sum_{r^{\prime}\in \{r\in \overline{S}: i\in V_r\}}{c_{{r^{\prime},1}}}$. 
    Hence, the NPI cost of this metric is  $\mathcal{C}_{1}(\overline{S}) = \sum_{r\in \overline{S}}c_{r,1}|V_r|$. 
    \item Maximum cost: Let $c_{r,2}$ denote the economic/social cost of each individual in cluster $r$ due to NPIs (e.g., lost wages due to remaining at home). If $i$ is included in multiple clusters with different levels of desire to socialize, its cost is decided by the maximum cost of a node in those clusters, i.e., cost of node $i$ is $\overline{c}_{i}(\overline{S})=\max{\{c_{r,2}: r\in \overline{S}, i\in V_r\}}$ and the NPI cost 
    is $\mathcal{C}_2(\overline{S}) = \sum_{i\in \overline{V}(\overline{S})}\overline{c}_{i}(\overline{S})$, 
    where $\overline{V}(\overline{S}) = \cup_{r\in \overline{S}}V_{r}$.
    \item Identical cost: We consider the expense of distributing masks to all the individuals in selected clusters. We assume masks have the identical price as $c_0$. Hence, if $i$ is included in any selected clusters, the cost of node $i$ is $c_0$. 
    The NPI cost of this metric is $\mathcal{C}_3(\overline{S}) = \sum_{i\in \overline{V}(\overline{S})}c_0$. 
\end{itemize}
Therefore, the total cost of NPI strategy is $\mathcal{C}(\overline{S}) = \alpha_1 \mathcal{C}_1(\overline{S}) + \alpha_2 \mathcal{C}_2(\overline{S}) + \alpha_3 \mathcal{C}_3(\overline{S}),$
where $\alpha_1$, $\alpha_2$, and $\alpha_3$ are nonnegative coefficients that are used to trade off these three metrics.

\subsection{Background on Submodularity}
Let $Z$ be a finite set. A function $f: 2^{Z} \rightarrow \mathbb{R}$ is submodular if, for any $X \subseteq Y \subseteq Z$ and $v \in Z\setminus Y$, we have $$f(X \cup \{v\}) - f(X) \geq f(Y \cup \{v\})-f(Y).$$ 
A function $f: 2^{Z} \rightarrow \mathbb{R}$ is monotone nonincreasing (resp. nondecreasing) if, for any sets $X, Y$ with $X \subseteq Y$, we have $f(X) \geq f(Y)$ (resp. $f(Y) \leq f(X)$). A function $f$ is supermodular if $-f$ is submodular. Nonnegative weighted sums of supermodular (resp. submodular) functions are supermodular (resp. submodular) \cite{fujishige2005submodular}. Furthermore, if $f(X)$ is nonincreasing and supermodular, then $\max{\{f(X),c\}}$ is nonincreasing and supermodular for any real number $c$ \cite{fujishige2005submodular}.
\section{Endemic State Minimization}
\label{sec:endemic}
This section formulates the problem of choosing an optimal cluster set to apply the NPI strategy for a population. We first give the problem formulation and prove it is equivalent to a submodular covering problem, and then we present our framework for solving the problem.

In the present paper, we consider a scenario in which $x_{i}(0)\in [0,1]$, $\forall i\in V$ with $\boldsymbol{x}(0)\neq \boldsymbol{0}$, and the system (\ref{eq:system-dynamics}) has an endemic state which is globally asymptotically stable \cite{khanafer2016stability}.
Specifically, we study the problem of selecting clusters $\overline{S}$ to apply NPI to minimize the cost of this strategy while ensuring that the steady state, either disease-free state or endemic state, is lower than a given bound.

Based on the system dynamics in (\ref{eq:CT-node-dynamics}), we have for all $i\in V$, the steady state $\boldsymbol{x}$ satisfies the condition $-\gamma_{i} x_{i}+(1-x_{i})\sum_{j\in N(i)}\lambda_{ij}(\overline{S})x_j=0,$
where $\lambda_{ij}(\overline{S}) = \beta_{i}a_{ij}(\overline{S})$.
Define the function $J_{i}(\overline{S};\boldsymbol{x})$ as
$$J_{i}(\overline{S};\boldsymbol{x}) = -\gamma_i x_{i}+(1-x_{i})\sum_{j\in N(i)}\lambda_{ij}(\overline{S})x_j.$$ Therefore, given a desired endemic state $\boldsymbol{\hat{x}}>0$, the problem can be formulated as
\begin{equation}
\label{eq:quasi-stationary-problem}
\begin{array}{ll}
\underset{\overline{S}}{\mbox{minimize}} & \mathcal{C}(\overline{S})\\
\mbox{s.t.} & \boldsymbol{\hat{x}} \geq \boldsymbol{x}^{\ast}\\
&J_{i}(\overline{S};\boldsymbol{x}^{\ast})=0, \forall i\in V,
\end{array}
\end{equation}
where $\boldsymbol{x}^{\ast}$ denotes the steady state under the strategy $\overline{S}$.  
Our approach to solving (\ref{eq:quasi-stationary-problem}) will first prove a monotonicity condition, namely, if we relax the constraints of (\ref{eq:quasi-stationary-problem}) to 
\begin{equation}
\label{eq:relaxed-condition}
J_{i}(\overline{S};\boldsymbol{\hat{x}})\leq0, \forall i\in V, 
\end{equation}
then any NPI strategy satisfying (\ref{eq:relaxed-condition}) results in 
steady state $\boldsymbol{x}^{\ast}$ that is element-wise bounded by $\boldsymbol{\hat{x}}$. We will then prove that the relaxed constraint (\ref{eq:relaxed-condition}) is equivalent to a submodular function constraint. Let $\lambda(\overline{S}) = [\lambda_{ij}(\overline{S})]$, we have $\mathcal{R}_{0}(\overline{S}) = \rho(D^{-1}\lambda(\overline{S}))$ in system \eqref{eq:DT-system-dynamics}.
We have the following preliminary result \cite[Theorem 5]{van2008virus}. 
\begin{theorem}[\cite{van2008virus}]
\label{theorem:continued-fraction}
Suppose that $\mathcal{R}_{0}(\overline{S})\geq 1$. The nonzero steady-state infection probability of any node $i$ can be expressed as a continued fraction
\begin{equation}
    \label{eq:continued-fraction}
    {x}^{\ast}_{i} = 1 - \frac{1}{1+\mu_{i} d_{i} - \mu_{i} \sum_{j \in N(i)}{\frac{\lambda_{ij}(\overline{S})}{1+\mu_{j} d_{j} - \mu_{j}\sum_{k \in N(j)}{\frac{\lambda_{jk}(\overline{S})}{\ddots}}}}},
\end{equation}
where $\mu_{i}=\frac{1}{\gamma_{i}}$ and $d_{i} = \sum_{j \in N(i)}{\lambda_{ij}(\overline{S})}$.
\end{theorem}
We use $\boldsymbol{x}^{\ast}(\lambda)$ to denote the steady state for the system arising from the set of propagation rates $\lambda(\overline{S})$. We use $\lambda$ to denote $\lambda(\overline{S})$ for simplicity. Define $\overline{\lambda} = [\overline{\lambda}_{ij}]$ with $\overline{\lambda}_{ij}=\beta_{i}\alpha_{ij}\epsilon_{ij}$, where $\epsilon_{ij}>0$. We have the following result.
\begin{theorem}
    \label{theorem:monotonicity}
    Suppose that $\lambda \leq \overline{\lambda}$. Then, ${\boldsymbol{x}}^{\ast}(\lambda) \leq {\boldsymbol{x}}^{\ast}(\overline{\lambda})$.
\end{theorem}

\begin{proof}
We have $D^{-1}{\lambda}$ and $D^{-1}\overline{\lambda}$ are both nonnegative matrices with $0\leq D^{-1}{\lambda}\leq D^{-1}{\overline{\lambda}}$. 
By \cite[Corollary 8.1.19]{horn2012matrix}, we then have $\rho(D^{-1}{\lambda})\leq \rho(D^{-1}{\overline{\lambda}})$. We divide the proof into three cases: (i) $\rho(D^{-1}{\overline{\lambda}})\leq 1$, (ii) $\rho(D^{-1}{\lambda})\leq 1 < \rho(D^{-1}{\overline{\lambda}})$, and (iii) $1<\rho(D^{-1}{\lambda})\leq \rho(D^{-1}{\overline{\lambda}})$.

Case (i): We have $\rho(D^{-1}{\lambda}) \leq \rho(D^{-1}{\overline{\lambda}})\leq 1$ implies $\boldsymbol{x}^{\ast}({\lambda})= \boldsymbol{x}^{\ast}(\overline{\lambda})=0$.

Case (ii): We have $\rho(D^{-1}{\lambda})\leq 1 < \rho(D^{-1}{\overline{\lambda}})$ implies $\boldsymbol{x}^{\ast}({\lambda})=0< \boldsymbol{x}^{\ast}(\overline{\lambda})$.

Case (iii): We have $1<\rho(D^{-1}{\lambda})\leq \rho(D^{-1}{\overline{\lambda}})$ implies $\boldsymbol{x}^{\ast}({\lambda})>0$ and $\boldsymbol{x}^{\ast}(\overline{\lambda})>0$.

To prove ${\boldsymbol{x}}^{\ast}(\lambda) \leq {\boldsymbol{x}}^{\ast}(\overline{\lambda})$, we first define a sequence of functions $\phi_{i}^{l}(\lambda)$ for $i=1,\ldots,n$ and $l=1,2,\ldots$ by $\phi_{i}^{1}(\lambda) = 1 + \mu_{i} d_{i}$ and $\phi_{i}^{l}(\lambda) = 1 + \mu_{i} d_{i} - \mu_{i}\sum_{j \in N(i)}{\frac{\lambda_{ij}}{\phi_{j}^{l-1}(\lambda)}}.$ By Theorem \ref{theorem:continued-fraction}, we have 
\begin{equation}
    \label{eq:ss-limit}
    {x}^{\ast}_{i}(\lambda) = \lim_{l \rightarrow \infty}{\left(1-\frac{1}{\phi_{i}^{l}(\lambda)}\right)}.
\end{equation}

We will first show that, for all $i,l$, $\phi_{i}^{l}(\lambda)$ is a nondecreasing function in $\lambda$. The proof is by induction. We have that $\phi_{i}^{1}(\lambda) = 1 + \mu_{i}\sum_{j \in N(i)}{\lambda_{ij}},$ which is clearly nondecreasing in $\lambda$. Proceeding inductively, 
\begin{eqnarray*}
    \phi_{i}^{l}(\lambda) &=& 1 + \mu_{i}\sum_{j \in N(i)}{\lambda_{ij}} - \mu_{i}\sum_{j \in N(i)}{\frac{\lambda_{ij}}{\phi_{j}^{l-1}(\lambda)}} \\
    &=& 1 + \mu_{i}\sum_{j \in N(i)}{\lambda_{ij}\left(1 - \frac{1}{\phi_{j}^{l-1}(\lambda)}\right)}.
\end{eqnarray*}
Since $\phi_{j}^{l-1}(\lambda)$ is nondecreasing in $\lambda$, we have that $\left(1 - \frac{1}{\phi_{j}^{l-1}(\lambda)}\right)$ is nondecreasing in $\lambda$. Furthermore, $\lambda_{ij}$ is nondecreasing in $\lambda$, and hence $\phi_{i}^{l}$ is nondecreasing in $\lambda$. We thus have (by (\ref{eq:ss-limit}), ${x}_{i}^{\ast}(\lambda)$ is a pointwise limit of nondecreasing functions of $\lambda$, and hence is nondecreasing in $\lambda$, implying ${\boldsymbol{x}}^{\ast}(\lambda) \leq {\boldsymbol{x}}^{\ast}(\overline{\lambda})$, completing the proof.
\end{proof}

With the result in Theorem \ref{theorem:monotonicity}, we prove (\ref{eq:relaxed-condition}) is a sufficient condition for steady state $\boldsymbol{x}^{\ast}(\lambda(\overline{S}))$ to be bounded below $\boldsymbol{\hat{x}}$.
\begin{proposition}
\label{prop:spreading-sufficient-condition}
Suppose that, for a given $\hat{\boldsymbol{x}}$, the NPI strategy $\overline{S}$ satisfies $J_{i}(\overline{S},\boldsymbol{\hat{x}}) \leq 0$ for all $i \in V$. Then $\boldsymbol{x}^{\ast}(\lambda(\overline{S})) \leq \hat{\boldsymbol{x}}$.
\end{proposition}

\begin{proof}
    For each $i$ with $J_{i}(\overline{S};\hat{\boldsymbol{x}}) = 0$, we let $\tilde{\lambda}_{ij}=\lambda_{ij}(\overline{S})$.
    For each $i$ with $J_{i}(\overline{S};\hat{\boldsymbol{x}}) < 0$, we let $\tilde{\lambda}_{ij} \geq \lambda_{ij}(\overline{S})$ such that 
    $-\gamma_i \hat{x}_{i}+(1-\hat{x}_{i})\sum_{j\in N(i)}\tilde{\lambda}_{ij}\hat{x}_{i}=0.$ We have $\lambda(\overline{S}) \leq \tilde{\lambda}$ and $\boldsymbol{x}^{\ast}(\tilde{\lambda}) = \hat{\boldsymbol{x}}$. Thus $\boldsymbol{x}^{\ast}(\lambda(\overline{S})) \leq \boldsymbol{x}^{\ast}(\tilde{\lambda}) = \boldsymbol{\hat{x}}$.
\end{proof}
Therefore, we can recast problem (\ref{eq:quasi-stationary-problem}) as
\begin{equation}
\label{eq:relax-quasi-stationary-problem}
\begin{array}{ll}
\underset{\overline{S}}{\mbox{minimize}} & 
\mathcal{C}(\overline{S})\\
\mbox{s.t.} & \sum_{i=1}^{N} \max\{J_{i}(\overline{S};\boldsymbol{\hat{x}}), 0\}\leq 0
\end{array}
\end{equation}
We next prove the constraint in (\ref{eq:relax-quasi-stationary-problem}) has a submodular structure. 
\begin{lemma}
\label{lemma:lambda}
    The function $\lambda_{ij}(\overline{S})$ is a positive nonincreasing and supermodular function in $\overline{S}$.
\end{lemma}
\begin{proof}
Assume we have $\overline{S}_{1}\subseteq \overline{S}_{2}\subseteq S$ with $w\in S\setminus \overline{S}_{2}$. We define 
$\overline{V}_{1} = \overline{V}(\overline{S}_1)$ and $\overline{V}_{2} = \overline{V}(\overline{S}_2)$.
We have $i^{\prime}\in \overline{V}_1$ implies $i^{\prime}\in \overline{V}_2$. 
Let $\Psi_{ij}(\overline{S}_{k}, \{w\})\triangleq \lambda_{ij}(\overline{S}_{k}\cup\{w\})-\lambda_{ij}(\overline{S}_{k})$. For each $i\in V_w$ and $j\in N(i)$, if $i\in \overline{V}_1$, then $\Psi_{ij}(\overline{S}_{2}, \{w\})= \Psi_{ij}(\overline{S}_{1}, \{w\})= 0$.
If $i\in \overline{V}_2\setminus \overline{V}_1$, we have $\Psi_{ij}(\overline{S}_{2}, \{w\})=0$, and 
\begin{align*}
        &\Psi_{ij}(\overline{S}_{1}, \{w\})=\left\{
        \begin{array}{ll}
        -\beta_{i}a_{ij}(\theta_2-\theta_1), & j\in \overline{V}_1\\
        -\beta_{i}a_{ij}\theta_1, & j\notin \overline{V}_1
        \end{array}
        \right.   
    \end{align*}
implying $\Psi_{ij}(\overline{S}_{1}, \{w\})\leq \Psi_{ij}(\overline{S}_{2}, \{w\})$.
If $i\notin \overline{V}_2$,
\begin{multline*}
        \Psi_{ij}(\overline{S}_{1}, \{w\}) - \Psi_{ij}(\overline{S}_{2}, \{w\})\\
        =\left\{
        \begin{array}{ll}
        -\beta_{i}a_{ij}(-\theta_2+2\theta_1), & j\in \overline{V}_2\setminus\overline{V}_1\\
        0, & \mbox{otherwise}\\
        \end{array}
        \right.   
    \end{multline*}
implying $\Psi_{ij}(\overline{S}_{1}, \{w\}) - \Psi_{ij}(\overline{S}_{2}, \{w\})\leq 0$. The function $\lambda_{ij}(\overline{S})$ is a supermodular function in $\overline{S}$. 
Furthermore, define $\eta_{ij}(\overline{S}_2, \overline{S}_1)\triangleq \lambda_{ij}(\overline{S}_{2}) - \lambda_{ij}(\overline{S}_{1})$. For $i\in \overline{V}_2\setminus\overline{V}_1$, $j\in N(i)$,
\begin{multline*}
        \eta_{ij}(\overline{S}_2, \overline{S}_1)
        =
        \left\{
        \begin{array}{ll}
        -\beta_{i}a_{ij}(\theta_2-\theta_1), & j\in \overline{V}_1\\
        -\beta_{i}a_{ij}\theta_2, & j\in \overline{V}_2\setminus\overline{V}_1\\
        -\beta_{i}a_{ij}\theta_1, & j\notin \overline{V}_2
        \end{array}
        \right.   
    \end{multline*}
implying the function $\lambda_{ij}(\overline{S})$ is a nonincreasing function in $\overline{S}$, completing the proof.
\end{proof}
We need an additional preliminary result as follows.
\begin{lemma}
\label{lemma:J}
    The function $J_{i}(\overline{S};\boldsymbol{\hat{x}})$ is a nonincreasing and supermodular function in $S$.
\end{lemma}
\begin{proof}
     We first prove $\sum_{j\in N(i)}\lambda_{ij}(\overline{S})\hat{x}_j$ is nonnegative, nonincreasing and supermodular in $\overline{S}$. Based on the results in Lemma \ref{lemma:lambda} and the fact that $\hat{x}_j\geq 0, \forall j\in N(i)$, we have  $\sum_{j\in N(i)}\lambda_{ij}(\overline{S})\hat{x}_j$ is a nonnegative weighted sums of nonincreasing modular functions, implying $\sum_{j\in N(i)}\lambda_{ij}(\overline{S})\hat{x}_j$ is a nonnegative, nonincreasing and supermodular function. 

    Hence, we have $J_{i}(\overline{S};\boldsymbol{\hat{x}})$ is a sum of a constant number $-\gamma_{i}\hat{x}_{i}$ and nonnegative weighted sums of nonnegative, nonincreasing supermodular function, implying $J_{i}(\overline{S};\boldsymbol{\hat{x}})$ is a nonincreasing and supermodular function in $\overline{S}$.
\end{proof}

We define $\overline{J}(\overline{S};\boldsymbol{\hat{x}}) = \sum_{i=1}^{N} \max\{J_{i}(\overline{S};\boldsymbol{\hat{x}}), 0\}$. Finally, we prove the constraint in (\ref{eq:relax-quasi-stationary-problem}) has a submodular structure.
\begin{proposition}
\label{prop:qs-cost-function}
    The function $\overline{J}(\overline{S};\boldsymbol{\hat{x}})$ is nonincreasing and supermodular function in $\overline{S}$.
\end{proposition}
\begin{proof}
    Based on Lemma \ref{lemma:J}, we have $\max\{J_{i}(\overline{S};\boldsymbol{\hat{x}}), 0\}$ is a nonincreasing and supermodular function in $\overline{S}$. Therefore, $\sum_{i=1}^{N} \max\{J_{i}(\overline{S};\boldsymbol{\hat{x}}), 0\}$ is a nonnegative weighted sum of nonincreasing supermodular functions, implying it is a nonincreasing supermodular function in $\overline{S}$.
\end{proof}
Next, we consider the structure of the cost function $\mathcal{C}(\overline{S})$.
\begin{lemma}
    For any $\overline{S}\subseteq S$, we have $\mathcal{C}(\overline{S})$ is a nondecreasing and submodular function of $\overline{S}$.
\end{lemma}
\begin{proof}
    Assume we have ${S}_{1}\subseteq \overline{S}_{2}\subseteq S$ with $w\in S\setminus \overline{S}_{2}$.
    We first prove  $\mathcal{C}_{1}(\overline{S})$ is a nondecreasing modular function of $\overline{S}$. $\mathcal{C}_1(\overline{S}) = \sum_{r\in \overline{S}}c_{r,1}|V_r|$. For $r\in S\setminus \overline{S}$, $\mathcal{C}_1(\overline{S}\cup\{r\}) - \mathcal{C}_1(\overline{S}) = {c_{r,1}|V_r|}$ implies that it is a nondecreasing function of $\overline{S}$. Furthermore, $\mathcal{C}_1(\overline{S}_1\cup\{w\}) - \mathcal{C}_1(\overline{S}_1) = \mathcal{C}_1(\overline{S}_2\cup\{w\}) - \mathcal{C}_1(\overline{S}_2) = c_{w,1}|V_w|$, implies it is a modular function of $\overline{S}$.   
    Next, we prove $\mathcal{C}_2(\overline{S})$ is a nondecreasing and submodular function of $\overline{S}$. We have $\overline{V}(\overline{S})\subseteq \overline{V}(\overline{S}\cup \{r\})$. Then, 
    \begin{multline*}
       \overline{c}_{i}(\overline{S}\cup\{r\})-\overline{c}_{i}(\overline{S}) \\
       = \left\{
        \begin{array}{ll}       
        0, & i\in \overline{V}(\overline{S})\setminus V_r\\
        \max\{c_{r,2} - \overline{c}_{i}(\overline{S}),0\}, & i\in \overline{V}(\overline{S})\cap V_r\\
        c_{r,2}, & i\in \overline{V}(\overline{S}\cup\{r\})\setminus \overline{V}(\overline{S})
        \end{array}
        \right.   
    \end{multline*}
    implying $\overline{c}_i(\overline{S})$ is a nondecreasing function of $\overline{S}$. Furthermore, define $\tilde{c}_{i}(\overline{S},r) = \max\{c_{r,2} - \overline{c}_{i}(\overline{S}),0\}$, we have 
    \begin{multline*}
      (\overline{c}_{i}(\overline{S}_1\cup\{w\})-\overline{c}_{i}(\overline{S}_1))- (\overline{c}_{i}(\overline{S}_2\cup\{w\})-\overline{c}_{i}(\overline{S}_2))=\\
      \left\{
        \begin{array}{ll}
        \tilde{c}_{i}(\overline{S}_1,w) - \tilde{c}_{i}(\overline{S}_2,w), & i\in \overline{V}(\overline{S}_1)\cap V_r\\
        c_{w,2}-\tilde{c}_i(\overline{S}_2,w), & i\in (\overline{V}(\overline{S}_2)\setminus \overline{V}(\overline{S}_1))\cap V_w\\
        0, & \mbox{otherwise}
        \end{array}
        \right.   
    \end{multline*}
    We can observe that $c_{w,2}-\overline{c}_i(\overline{S}_2)\leq c_{w,2}-\overline{c}_i(\overline{S}_1)$ because $\overline{c}_i(\overline{S})$ is nondecreasing. Hence, $\tilde{c}_{i}(\overline{S}_1,w) - \tilde{c}_{i}(\overline{S}_2,w)\geq 0$. Also, $\overline{c}_i(\overline{S})$ is nonnegative, implying $c_{w,2}-\overline{c}_{i}(\overline{S}_2,w)\geq 0$. Therefore, we have  $\mathcal{C}_2(\overline{S}) = \sum_{i\in \overline{V}(\overline{S})}\overline{c}_{i}(\overline{S})$ is a submodular function of $\overline{S}$. We observe that when we set $c_r=c_0$, $\forall r\in S$, then $\mathcal{C}_3(\overline{S})$ = $\mathcal{C}_2(\overline{S})$. Therefore, $\mathcal{C}_3(\overline{S})$ is a nondecreasing and submodular function of $\overline{S}$. Therefore, $\mathcal{C}(\overline{S})$ is a nonnegative weighted sums of 
    nondecreasing and submodular functions of $\overline{S}$, implying $\mathcal{C}(\overline{S})$ is a nondecreasing and submodular function of $\overline{S}$,
    completing the proof.
\end{proof}

Based on Proposition \ref{prop:qs-cost-function}, we have \eqref{eq:relax-quasi-stationary-problem} is a submodular cost submodular cover (SCSC) problem, which can be solved by the algorithm proposed in \cite{iyer2013submodular}. The iterative algorithms in \cite{iyer2013submodular} select appropriate surrogate functions to the cost function and constraint function at each iteration and optimize over them until convergence or a predefined iteration. Algorithm optimality bounds follow from results in \cite{iyer2013submodular} and are omitted due to space constraints. When we consider $\mathcal{C}(\overline{S}) = \mathcal{C}_1(\overline{S})$,
Eq. \eqref{eq:relax-quasi-stationary-problem} is a submodular covering problem and can be solved by the greedy algorithm in polynomial time as follows \cite{wolsey1982analysis} . Assume there exists a $\hat{S}\subseteq S$ such that $\overline{J}(\hat{S} ; \hat{\boldsymbol{x}}) = 0$. The set $\overline{S}$ is initialized to be empty. At each iteration, select an element $r \in S \setminus \overline{S}$ that minimize $\frac{\overline{J}(\overline{S} \cup \{r\}; \boldsymbol{\hat{x}})-\overline{J}(\overline{S}; \boldsymbol{\hat{x}})}{c_{r,1}|V_r|},$ and add this element to $\overline{S}$. The algorithm terminates when $\overline{J}(\overline{S} ; \hat{\boldsymbol{x}}) = 0$. The following proposition describes the optimality guarantees of this approach. 
\begin{proposition}
    \label{prop:qs-optimality}
    Suppose there is an NPI strategy $\overline{S}$ satisfying $\boldsymbol{x}^{\ast}(\lambda(\overline{S})) \leq \boldsymbol{\hat{x}}$. The greedy algorithm for selecting an NPI strategy returns a set $\overline{S}$ 
    satisfying $\frac{\mathcal{C}_1(\overline{S})}{\mathcal{C}_1(S^{\ast})} \leq 1 + \log{\left\{\frac{\overline{J}(\emptyset ; \boldsymbol{\hat{x}})}{\overline{J}(\hat{S} ; \boldsymbol{\hat{x}})}\right\}},$ where $\hat{S}$ is the value of $\overline{S}$ at the second-last iteration of the algorithm,  $S^{\ast}$ is the optimal solution to \eqref{eq:relax-quasi-stationary-problem}.
\end{proposition}

The proof follows from the bounds on submodular covering in \cite{wolsey1982analysis} and the assumption that there exists $\overline{S}$ satisfying $\boldsymbol{x}^{\ast}(\lambda(\overline{S})) \leq \boldsymbol{\hat{x}}$ implying $\overline{J}(S; \mathbf{x}^{\ast}) = 0$. In the submodular covering problem, the greedy algorithm selects the element with highest gain per unit cost in each iteration with a polynomial time complexity. Proposition \ref{prop:qs-optimality} implies the greedy algorithm can find the near-optimal solution.

\section{Simulation}
\label{sec:simulation}

\begin{figure}[!bp]
\centering
\begin{subfigure}{0.23\textwidth}
    \centering
    \includegraphics[width = \textwidth]{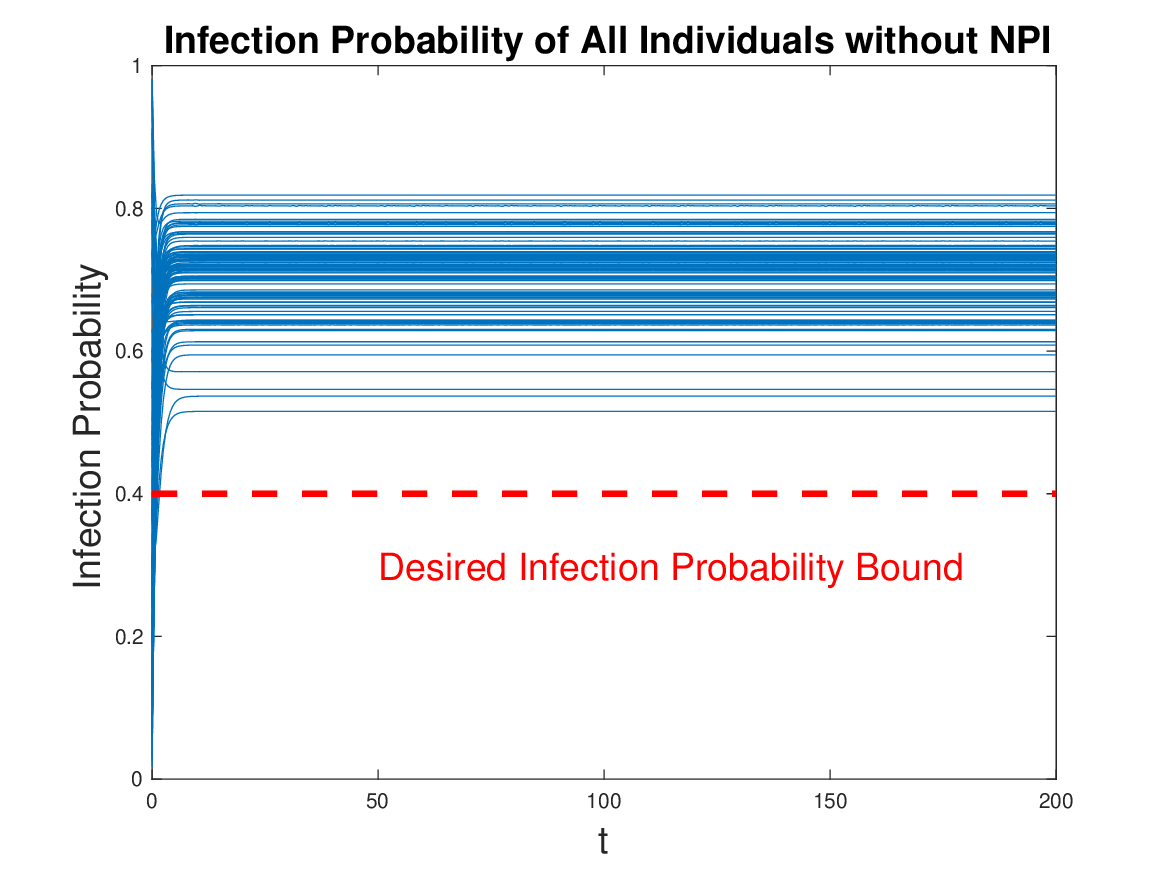}
    \subcaption{}
    \label{fig:pre-NPI}
\end{subfigure}
\hfill
\begin{subfigure}{0.23\textwidth}
  \centering
  \includegraphics[width = \textwidth]{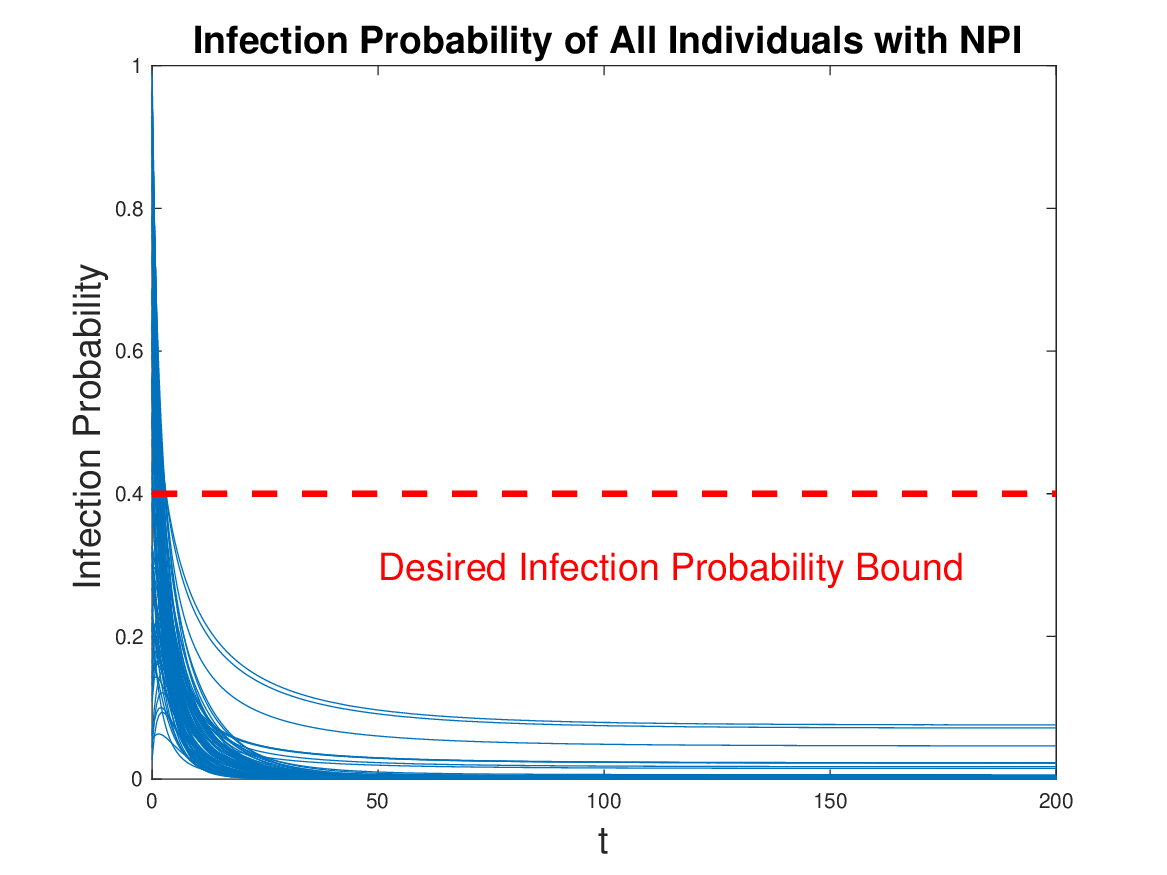}
  \subcaption{}
  \label{fig:post-NPI}
\end{subfigure}
\caption{\textcolor{black}{Infection probabilities of 100 individuals in the absence (left) and presence (right) of NPI strategies over time relative to the desired endemic state value (red dotted horizontal line). 
When NPI strategies are not implemented, the infection probability of individuals converges to the interval $[0.34, 0.73]$. 
All individuals in this case have a significantly higher infection probability than the desired value of $0.05$. 
We compute an NPI strategy by solving (\ref{eq:relax-quasi-stationary-problem}) using the greedy algorithm. 
Fig. \ref{fig:post-NPI} shows that when this NPI strategy is applied to $11$ clusters with $81$ nodes, infection probabilities converge to the interval $[0,\ 0.08]$, which is much lower than the desired endemic state of $0.4$. }
}
    \label{fig:sim}
\end{figure}
In this section, we demonstrate viability of our proposed NPI strategy by examining epidemic spread via an SIS model on a Watts-Strogatz small world network \cite{watts1998collective}.
The network we study has 100 nodes and 200 edges. We set $\gamma_{i}$ to be uniformly distributed in $[0.4, 0.5]$ and $\beta_{i}$ uniformly distributed in 
$ [0.6, 0.7]$ for all nodes, $a_{ij}$ uniformly distributed in
$[0.5, 0.7]$ for all edges, 
and the effect of NPIs as $\theta_1=0.7$ and $\theta_2=0.9$. 
The number of clusters is set to $25$, and number of nodes in each cluster is set between $10$ and $15$. The cost for each cluster is randomly selected from $\{1,2,3,4\}$. Nodes in each cluster are selected randomly. We consider the cost function $\mathcal{C}_1(\overline{S})$ 
and randomly initialize $x_{i}(0)\in [0,1]$ with $\boldsymbol{x}(0)\neq 0$. The desired infection probability is $\hat{x}_{i}=0.05$ for all nodes. In the simulation, we first check existence of a solution by checking if $\overline{J}(S;\boldsymbol{\hat{x}})\leq 0$. If a solution does not exist, we generate a new network. We compare our work to a baseline, in which the cluster with highest sum of node degrees is selected one by one until the constraint in (\ref{eq:relax-quasi-stationary-problem}) holds.

Fig. \ref{fig:sim} shows infection probabilities of 100 individuals in the absence (left) and presence (right) of NPI strategies over time. 
When NPI strategies are not implemented, the infection probability converges to the interval
$[0.52, 0.82]$ (Fig. \ref{fig:pre-NPI}). 
All individuals in this case have a 
higher infection probability than the desired value of $0.4$. 
We compute an NPI strategy by solving (\ref{eq:relax-quasi-stationary-problem}) using the greedy algorithm. 
Fig. \ref{fig:post-NPI} shows that when this strategy is applied to $11$ clusters with $81$ nodes, infection probabilities converge to [0,\ 0.08],
which is much lower than the desired endemic state of $0.4$. In Table  \ref{table: 0}, we set $\beta_{i}$ to be uniformly distributed in $[0.4, 0.6]$ for all nodes, $a_{ij}$ to be uniformly distributed in $[0.4, 0.5]$ for all edges, and maintain values of other parameters. We set the desired infection probability as $0.05$, $0.2$, $0.3, 0.4$, to compare the total NPI costs between the greedy algorithm and the baseline. For the same desired infection probability, the greedy algorithm selects $60$ to $66.7\%$ of clusters and $82.6$ to $86.5\%$ of nodes to satisfy constraints in \eqref{eq:relax-quasi-stationary-problem} compared to the baseline. The total NPI cost based on our greedy algorithm is $32.8$ to $49.1\%$ of that in the baseline.
\begin{table}[!htp]
\centering
\begin{tabular}{|l|c|c|c|c|c|}
\hline
                                                                            & \begin{tabular}[c]{@{}c@{}}Desired Infection\\ Probability\end{tabular} & 0.05 & 0.2 & 0.3 & 0.4 \\ \hline
\multirow{3}{*}{\begin{tabular}[c]{@{}l@{}}Greedy\\ Algorithm\end{tabular}} & \begin{tabular}[c]{@{}c@{}}Number of \\ selected clusters\end{tabular}  & 11   & 10  & 9   & 7   \\ \cline{2-6} 
                                                                            & \begin{tabular}[c]{@{}c@{}}Number of \\ selected nodes\end{tabular}     & 77   & 74  & 71  & 62  \\ \cline{2-6} 
                                                                            & NPI Cost                                                                & 303  & 236 & 195 & 147 \\ \hline
\multirow{3}{*}{Baseline}                                                   & \begin{tabular}[c]{@{}c@{}}Number of \\ selected clusters\end{tabular}  & 17   & 15  & 15 & 11  \\ \cline{2-6} 
                                                                            & \begin{tabular}[c]{@{}c@{}}Number of \\ selected nodes\end{tabular}     & 89   & 86  & 86  & 75  \\ \cline{2-6} 
                                                                            & NPI Cost                                                                & 617  & 554 & 554 & 448 \\ \hline
\end{tabular}
\caption{Comparison of total NPI cost computed by the greedy algorithm and the baseline. For the same desired infection probability, the greedy algorithm selects $60$ to $66.7$
$\%$ of clusters and $82.6$ to $86.5$
$\%$ of nodes to satisfy the constraints in \eqref{eq:relax-quasi-stationary-problem} compared with the baseline. The total NPI cost based on the greedy proposed algorithm is $32.8$ to $49.1$
$\%$ of the total NPI cost in the baseline.
}
\label{table: 0}
\end{table}
\section{Conclusion and future work}
\label{sec:conclusions}
This paper studied the problem of selecting low-cost strategies for NPIs to mitigate epidemics. 
We considered a networked SIS model relaxed by mean-field approximation. We assumed that the network includes multiple clusters and individuals in selected clusters take NPIs during interactions with their neighbors.
We showed that the problem of selecting a minimum-cost intervention strategy to ensure that the steady state of the infection probability remains below a specified threshold is submodular. 
Our experiments on a Watts-Strogatz small world graph showed that NPI strategies were highly effective in reducing infection levels of members of the population below a desired threshold. In the future, we plan to analyze sensitivity of our 
algorithm to network topologies.  
We will also examine cost functions that consider incentives based on insights from behavioral economics. 

\bibliographystyle{ieeetr}
\bibliography{bib_epidemic}
\end{document}